\documentclass[11pt]{article}

\usepackage{graphicx}
\usepackage{epsfig}
\usepackage{psfrag}
\usepackage{wrapfig}
\usepackage[all]{xy}

\usepackage[blocks]{authblk}

\usepackage[top=1in, bottom=1in, left=1in, right=1in]{geometry}

\usepackage{multirow}

\usepackage{setspace}
\usepackage{mathptmx}
\usepackage{url}

\usepackage{natbib}

\usepackage{amsmath}
\usepackage{amsfonts}
\usepackage{amsthm}
\usepackage{bbm}
\usepackage{array}

\newtheorem{prop}{Proposition}
\newtheorem{corollary}{Corollary}
\newtheorem{cond}{Condition}

\def\E{{\rm E}\,}

  \author[1,2]{Peter M. Aronow}
  \affil[1]{Department of Political Science, Yale University}
    \affil[2]{Department of Biostatistics, Yale School of Public Health}

  \author[2]{Forrest W. Crawford}

  \title{Nonparametric Identification for Respondent-Driven Sampling}

\date{\today}

\begin{document}


\maketitle


\begin{abstract}
\noindent Respondent-driven sampling is a survey method for hidden or hard-to-reach populations in which sampled individuals recruit others in the study population via their social links. The most popular estimator for for the population mean assumes that individual sampling probabilities are proportional to each subject's reported degree in a social network connecting members of the hidden population.  However, it remains unclear under what circumstances these estimators are valid, and what assumptions are formally required to identify population quantities.  In this short note we detail nonparametric identification results for the population mean when the sampling probability is assumed to be a function of network degree known to scale. Importantly, we establish general conditions for the consistency of the popular Volz-Heckathorn (VH) estimator.  Our results imply that the conditions for consistency of the VH estimator are far less stringent than those suggested by recent work on diagnostics for RDS. In particular, our results do not require random sampling or the existence of a network connecting the population. \\[1em]
\end{abstract}


\section{Introduction}

Respondent-driven sampling (RDS) is a method for surveying hidden or hard-to-reach populations such as sex workers or injection drug users \citep{Heckathorn1997Respondent,Broadhead1998Harnessing}.  Starting with a group of initial subjects called ``seeds'', respondents recruit others who are also members of the study population by giving them ``coupons" to present to the researcher.  These new subjects are interviewed, given coupons, and the process repeats.  Many researchers have approximated RDS as a sampling design in which the sampling probability for subject $i$ is proportional to their network degree $d_i$ \citep{Salganik2004Sampling,Volz2008Probability,Gile2010Respondent,Gile2011Improved}.  In particular, \citet{Salganik2004Sampling} and \citet{Volz2008Probability} justify this choice by modeling the recruitment process as a with-replacement random walk on a connected population network, where only one coupon is given to each subject, recruitment is uniformly at random from network neighbors, and each subject can be recruited infinitely many times.  
For an RDS sample of size $n$, \citet{Volz2008Probability} (hereafter VH) give the estimator 
\begin{equation}
  \hat{\mu}_{VH} = \frac{\sum_{i=1}^n y_i d_i^{-1}}{\sum_{i=1}^n d_i^{-1}}
  \label{eq:vh}
\end{equation}
where $y_i$ is the outcome of interest and $d_i$ is the degree of subject $i$.  

Several authors have expressed skepticism about RDS survey methodology in general and the VH estimator in particular \citep{Heimer2005Critical,Johnston2008Implementation,Goel2010Assessing,Gile2010Respondent,Salganik2012Commentary,White2012Respondent}. Many alternative characterizations of the recruitment process exist \citep{Goel2009Respondent,Gile2010Respondent,Gile2011Improved,Berchenko2013Modeling,Crawford2014Graphical}.  Empirical studies have also cast doubt on the performance of the VH estimator in real-world RDS datasets \citep{Wejnert2009Empirical,McCreesh2012Evaluation,Rudolph2013Importance}.  

A recent paper by \citet{Gile2015Diagnostics} presents diagnostics whose purpose is to help researchers determine whether the assumptions often invoked to motivate the VH estimator \eqref{eq:vh} are met in empirical RDS data. The diagnostics presented by  \citet{Gile2015Diagnostics} address a particular class of motivating assumptions about the structure of a hypothesized social network and the process by which new subjects are sampled. These assumptions, characterized by \citet[pg. 3]{Gile2015Diagnostics} as ``required by the [VH] estimator,''  are summarized in Table 1, reproduced from the original paper. 

\begin{table}
  \begin{tabular}{lcc}
    \hline
    & Network structure assumptions & Sampling assumptions \\
    \hline
  \multirow{2}{*}{Random-walk model} & \multirow{2}{*}{Network size large ($N\gg n$)} & With-replacement sampling, \\
                    &                               & single non-branching chain    \\
  \multirow{2}{*}{Remove seed dependence} & Homophily sufficiently weak, & \multirow{2}{*}{Enough sample waves} \\
                         & bottlenecks limited,       & \\
                         & connected graph & \\
    \multirow{2}{*}{Respondent behaviour} & \multirow{2}{*}{All ties reciprocated} & Degree accurately measured, \\
       & & random referral \\
    \hline
  \end{tabular}
  \label{tab:assumptions}
  \caption{Assumptions listed by \citet{Gile2015Diagnostics} as requirements for the VH estimator.  Reproduced from their Table 1. }
\end{table}

In this short note, we give an alternative, nonparametric set of conditions under which the VH estimator is consistent, and note identification conditions for a generalization of the VH estimator.   The conditions we articulate for consistency 
are restrictive and untestable, 
but they are nevertheless less stringent than the traditional model used to justify the VH estimator.  Consistency of the VH estimator does not require random sampling or even the existence of a network connecting the members of the study population. Our results clarify the inferential challenges posed by RDS data, challenges beyond those of other non-probability samples. Importantly, however, our results suggest conditions that can be more generally implied by other generative models that may justify the VH estimator or variants thereof.




\section{Results}

Formally, consider a sequence of populations and samples converging weakly to a joint limit distribution on the outcome, (reported) degree, and sample, denoted $(Y,D,S)$. Let the $\E[\cdot]$ and $\Pr[\cdot]$ operators refer to features of this limiting distribution.  In RDS, we observe the empirical joint distribution of the outcome $Y$ and degree $D$ conditional on the sampling indicator $S = 1$.  Without loss of generality, suppose that $Y$ has bounded support and that $D$ has support in the set $\{1,...,K\}$. 

\begin{cond}[Ignorability]
\label{cond:ignorability}
For all $k$ such that $\Pr[D = k] > 0$, $\E[Y | S = 1, D = k] = \E[Y | D = k]$ and $\Pr[S = 1 \vert D = k] > 0$.
\end{cond}

\begin{cond}[Knowledge of the Conditional Probability of Sampling]
\label{cond:samplingprob}
$\Pr[S = 1 \vert D = k] = f(k)$, where $f(\cdot)$ is known up to a unknown scale parameter $c$.
\end{cond}


\begin{prop} 
  \label{prop:identification}
  Given Conditions \ref{cond:ignorability} and \ref{cond:samplingprob}, the population mean is identified, with 
\begin{equation}
  \E[Y] =   \frac{\sum_{k=1}^K \E[Y | S = 1, D = k] \frac{\Pr[D = k | S = 1]}{f(k)}}{\sum_{k=1}^K \frac{\Pr[D = k | S = 1]}{f(k)}}.
\end{equation}
\end{prop}
\begin{proof}
  We can identify $\E[Y | D = k]$ for each degree $k$ from Condition \ref{cond:ignorability} since $\E[Y|D=k] = \E[Y|S=1,D =k]$.  We can identify each $\Pr[D = k]$ to scale directly from Condition \ref{cond:samplingprob} as 
\begin{equation}
  \Pr[D = k] = \frac{\Pr[D = k | S = 1]\Pr[S = 1]}{\Pr[S = 1 | D = k]} = \frac{\Pr[S=1]}{c} \frac{\Pr[D = k | S = 1]}{f(k)}.
\end{equation}
Then by the law of total expectation,
\begin{equation}
\E[Y] = \frac{\sum_{k=1}^K \E[Y | D = k]  \frac{\Pr[S=1]}{c} \frac{\Pr[D = k | S = 1]}{f(k)}}{\sum_{k=1}^K \frac{\Pr[S=1]}{c} \frac{\Pr[D = k | S = 1]}{f(k)}} =  \frac{\sum_{k=1}^K \E[Y | D = k] \frac{\Pr[D = k | S = 1]}{f(k)}}{\sum_{k=1}^K \frac{\Pr[D = k | S = 1]}{f(k)}}.
\end{equation}
\end{proof}
\noindent Given Proposition \ref{prop:identification}, consistency of the VH estimator directly follows from convergence of sample analogues to population quantities.
\begin{corollary} 
  Given Conditions \ref{cond:ignorability} and \ref{cond:samplingprob}, the VH estimator is consistent for $\E[Y]$ if $f(k) \propto k$.
  \end{corollary}
  
  \section{Discussion}

  A variant of Condition \ref{cond:ignorability} is usually assumed implicitly in statistical arguments in favor of the VH estimator \citep{Salganik2004Sampling,Salganik2006Variance,Volz2008Probability}. Ignorability is not empirically testable from RDS data alone, since researchers never observe $\E[Y|S=0,D=k]$ for any $k$. While ignorability is a strong but common assumption imposed for inference from non-probability samples, Condition \ref{cond:samplingprob} highlights the additional challenges posed by RDS data. The researcher does not generally have knowledge of the population distribution of degree, and thus ignorability with respect to degree is not sufficient to identify the population mean.  Specification of the conditional sampling probability in Condition \ref{cond:samplingprob} provides an alternative means for identification, and has typically been the focus of researchers' efforts to justify the VH estimator. The random-walk argument serves to motivate the choice of $f(k) \propto k$ in the VH estimator, but is not strictly necessary for its consistency. Under any model that implies subjects with higher reported degrees are more likely to be sampled and $f(k) \propto k$ characterizes this relationship, Condition \ref{cond:samplingprob} holds. Finally, we note that our results suggest that the VH estimator and variants thereof may be appropriate even when diagnostics predicated on a more restrictive model \citep[e.g.,][]{Gile2015Diagnostics} fail.

Without knowledge of the characteristics of the unsampled subjects, neither Condition 1 nor Condition 2 has directly testable implications, and thus the value of any diagnostics must depend on further assumptions about the generative process.  
Under further parametric assumptions, some of the conditions listed in Table 1 might be sufficient to imply consistency of the VH estimator.  A formalization of these assumptions as part of a generative model for the recruitment process would allow researchers to evaluate the statistical properties of diagnostics like those proposed by \citet{Gile2015Diagnostics}.

\bibliographystyle{spbasic}

\bibliography{results}

\begin{thebibliography}{19}
\providecommand{\natexlab}[1]{#1}
\providecommand{\url}[1]{{#1}}
\providecommand{\urlprefix}{URL }
\expandafter\ifx\csname urlstyle\endcsname\relax
  \providecommand{\doi}[1]{DOI~\discretionary{}{}{}#1}\else
  \providecommand{\doi}{DOI~\discretionary{}{}{}\begingroup
  \urlstyle{rm}\Url}\fi
\providecommand{\eprint}[2][]{\url{#2}}

\bibitem[{Berchenko et~al(2013)Berchenko, Rosenblatt, and
  Frost}]{Berchenko2013Modeling}
Berchenko Y, Rosenblatt J, Frost SD (2013) Modeling and analysing respondent
  driven sampling as a counting process. arXiv preprint arXiv:13043505

\bibitem[{Broadhead et~al(1998)Broadhead, Heckathorn, Weakliem, Anthony,
  Madray, Mills, and Hughes}]{Broadhead1998Harnessing}
Broadhead RS, Heckathorn DD, Weakliem DL, Anthony DL, Madray H, Mills RJ,
  Hughes J (1998) Harnessing peer networks as an instrument for {AIDS}
  prevention: results from a peer-driven intervention. Public Health Reports
  113(Suppl 1):42

\bibitem[{Crawford(2014)}]{Crawford2014Graphical}
Crawford FW (2014) The graphical structure of respondent-driven sampling. ArXiv
  Pre-print \urlprefix\url{http://arxiv.org/pdf/1406.0721}

\bibitem[{Gile(2011)}]{Gile2011Improved}
Gile KJ (2011) Improved inference for respondent-driven sampling data with
  application to {HIV} prevalence estimation. Journal of the American
  Statistical Association 106(493):135--146

\bibitem[{Gile and Handcock(2010)}]{Gile2010Respondent}
Gile KJ, Handcock MS (2010) Respondent-driven sampling: An assessment of
  current methodology. Sociological Methodology 40(1):285--327

\bibitem[{Gile et~al(2015)Gile, Johnston, and Salganik}]{Gile2015Diagnostics}
Gile KJ, Johnston LG, Salganik MJ (2015) Diagnostics for respondent-driven
  sampling. Journal of the Royal Statistical Society, Series A In press

\bibitem[{Goel and Salganik(2009)}]{Goel2009Respondent}
Goel S, Salganik MJ (2009) Respondent-driven sampling as {Markov} chain {Monte
  Carlo}. Statistics in Medicine 28(17):2202--2229

\bibitem[{Goel and Salganik(2010)}]{Goel2010Assessing}
Goel S, Salganik MJ (2010) Assessing respondent-driven sampling. Proceedings of
  the National Academy of Sciences 107(15):6743--6747

\bibitem[{Heckathorn(1997)}]{Heckathorn1997Respondent}
Heckathorn DD (1997) Respondent-driven sampling: a new approach to the study of
  hidden populations. Social Problems 44(2):174--199

\bibitem[{Heimer(2005)}]{Heimer2005Critical}
Heimer R (2005) Critical issues and further questions about respondent-driven
  sampling: comment on {Ramirez-Valles}, et al.(2005). AIDS and Behavior
  9(4):403--408

\bibitem[{Johnston et~al(2008)Johnston, Malekinejad, Kendall, Iuppa, and
  Rutherford}]{Johnston2008Implementation}
Johnston LG, Malekinejad M, Kendall C, Iuppa IM, Rutherford GW (2008)
  Implementation challenges to using respondent-driven sampling methodology for
  {HIV} biological and behavioral surveillance: field experiences in
  international settings. AIDS and Behavior 12(1):131--141

\bibitem[{McCreesh et~al(2012)McCreesh, Frost, Seeley, Katongole, Tarsh,
  Ndunguse, Jichi, Lunel, Maher, Johnston et~al}]{McCreesh2012Evaluation}
McCreesh N, Frost S, Seeley J, Katongole J, Tarsh MN, Ndunguse R, Jichi F,
  Lunel NL, Maher D, Johnston LG, et~al (2012) Evaluation of respondent-driven
  sampling. Epidemiology 23(1):138

\bibitem[{Rudolph et~al(2013)Rudolph, Fuller, and
  Latkin}]{Rudolph2013Importance}
Rudolph AE, Fuller CM, Latkin C (2013) The importance of measuring and
  accounting for potential biases in respondent-driven samples. AIDS and
  Behavior 17(6):2244--2252

\bibitem[{Salganik(2006)}]{Salganik2006Variance}
Salganik MJ (2006) Variance estimation, design effects, and sample size
  calculations for respondent-driven sampling. Journal of Urban Health
  83(1):98--112

\bibitem[{Salganik(2012)}]{Salganik2012Commentary}
Salganik MJ (2012) Commentary: respondent-driven sampling in the real world.
  Epidemiology 23(1):148--150

\bibitem[{Salganik and Heckathorn(2004)}]{Salganik2004Sampling}
Salganik MJ, Heckathorn DD (2004) Sampling and estimation in hidden populations
  using respondent-driven sampling. Sociological Methodology 34(1):193--240

\bibitem[{Volz and Heckathorn(2008)}]{Volz2008Probability}
Volz E, Heckathorn DD (2008) Probability based estimation theory for respondent
  driven sampling. Journal of Official Statistics 24(1):79--97

\bibitem[{Wejnert(2009)}]{Wejnert2009Empirical}
Wejnert C (2009) An empirical test of respondent-driven sampling: Point
  estimates, variance, degree measures, and out-of-equilibrium data.
  Sociological Methodology 39(1):73--116

\bibitem[{White et~al(2012)White, Lansky, Goel, Wilson, Hladik, Hakim, and
  Frost}]{White2012Respondent}
White RG, Lansky A, Goel S, Wilson D, Hladik W, Hakim A, Frost SD (2012)
  Respondent driven sampling—where we are and where should we be going?
  Sexually Transmitted Infections 88(6):397--399

\end{thebibliography}

\end{document}